\documentclass[11pt,letterpaper]{article}

\usepackage[margin=1in]{geometry}
\usepackage{cite}
\usepackage{hyperref}
\usepackage{colortbl}
\usepackage{paralist}
\usepackage{pdfpages}
\usepackage{enumitem}
\usepackage{float}
\usepackage{booktabs}
\usepackage[override]{cmtt}
\usepackage{color,xcolor}
\usepackage{graphicx,eso-pic}
\usepackage{boxedminipage}
\usepackage{url}
\usepackage{caption}
\usepackage{subcaption}
\usepackage{xspace}
\usepackage{multirow}
\usepackage{amsmath,amsthm,amstext,amssymb,amsfonts,latexsym}
\usepackage{wrapfig}
\usepackage[linesnumbered,ruled,vlined]{algorithm2e}
\usepackage[noend]{algpseudocode}
\usepackage{algorithmicx}
\usepackage{epstopdf}
\usepackage{mdframed}
\usepackage{cases}
\usepackage[capitalize]{cleveref}
\usepackage[compact]{titlesec}

\newcommand{\poly}{{\normalfont \texttt{poly}}}
\newcommand{\one}{\ensuremath{\mathbf{1}}} 

\DeclareMathOperator{\Diag}{Diag}

\newcommand{\norm}[1]{\lVert#1\rVert}

\newcounter{note}[section]

\newcommand{\yuetodo}[1]{{\large\color{green}[Yue todo: #1]}}

\definecolor{yue}{rgb}{0.7, 0, 0}

\renewcommand{\yuetodo}[1]{}

\newcommand{\mcal}[1]{\ensuremath{\mathcal {#1}}}

\definecolor{darkgreen}{rgb}{0,0.5,0}
\definecolor{lightblue}{RGB}{0,176,240}
\definecolor{darkblue}{RGB}{0,112,192}
\definecolor{lightpurple}{RGB}{124, 66, 168}
\definecolor{grey}{RGB}{139, 137, 137}
\definecolor{maroon}{RGB}{178, 34, 34}
\definecolor{green}{RGB}{34, 139, 34}
\definecolor{types}{RGB}{72, 61, 139}
\definecolor{gold}{rgb}{0.8, 0.33, 0.0}

\definecolor{darkgray}{gray}{0.3}

\definecolor{darkred}{rgb}{0.5, 0, 0}
\definecolor{darkgreen}{rgb}{0, 0.5, 0}
\definecolor{darkblue}{rgb}{0,0,0.5}

\newcommand\markx[2]{}

\newcommand{\Z}{\mathbb{Z}}

\newcommand{\R}{\mathbb{R}}

\newcommand{\ignore}[1]{}

\newcounter{task}

\newenvironment{proofof}[1]{\vspace{10pt}
\noindent \textbf{Proof of #1:}}{\hfill \qed}

\newtheorem{theorem}{Theorem}[section]

\newtheorem{corollary}[theorem]{Corollary}
\newtheorem{fact}[theorem]{Fact}
\newtheorem{lemma}[theorem]{Lemma}

{
\theoremstyle{definition}
\newtheorem{definition}[theorem]{Definition}
}

\newcounter{cnt:challenge}

\SetKwComment{Comment}{//}{}

\begin{document}

\title{On the Hardness of Opinion Dynamics Optimization with $L_1$-Budget on Varying Susceptibility to Persuasion}


\author{T-H. Hubert Chan\thanks{The University of Hong Kong. \texttt{hubert@cs.hku.hk}} \and
Chui Shan Lee\thanks{The University of Hong Kong. \texttt{leechuishan@connect.hku.hk}}}

\date{}

\begin{titlepage}

\maketitle

\begin{abstract}
Recently, Abebe et al. (KDD 2018) and Chan et al. (WWW 2019) have considered
an opinion dynamics optimization problem that is based on
a popular model for social opinion dynamics, in which each agent has some fixed innate opinion, and a resistance that measures the importance it places on its innate opinion; moreover, the agents influence one another's opinions through an iterative process. Under certain conditions, this iterative process converges to some equilibrium opinion vector.  
Previous works gave an efficient local search algorithm
to solve the unbudgeted variant of the problem, for which the goal is to modify the resistance of any number of agents (within some given range) such that the sum of the equilibrium opinions is minimized.
On the other hand, it was proved that the $L_0$-budgeted variant is NP-hard,
where the $L_0$-budget is a restriction given upfront on the number of agents
whose resistance may be modified.

Inspired by practical situations in which the effort to modify an agent's
resistance increases with the magnitude of the change,
we propose the $L_1$-budgeted variant, in which
the $L_1$-budget is a restriction on the sum of the magnitudes
of the changes over all agents' resistance parameters.
In this work, we show that the $L_1$-budgeted variant is NP-hard via a reduction
from vertex cover.  However, contrary to the $L_0$-budgeted variant,
a very technical argument is needed to show that the optimal solution
can be achieved by focusing the given $L_1$-budget on as small a number of agents as possible,
as opposed to spreading the budget over a large number of agents.

\end{abstract}

\thispagestyle{empty}
\end{titlepage}

\section{Introduction}
\label{sec:intro}

The process of social influence is a significant basis for opinion formation, decision making and the shaping of an individual's identity. It drives many social phenomenon ranging from the emergence of trends, diffusion of rumor, and the shaping of public views about social issues.
An opinion formation model was introduced by the works of DeGroot~\cite{DeGroot1974} and Friedkin and Johnsen~\cite{Friedkin1999}, which considered how agents' influence one another's opinions in discrete time steps.
In this model, each agent~$i$ has some {\em innate opinion} $s_i$ 
in $[0,1]$, which
reflects the intrinsic position of the agent on a certain topic.
The expressed opinion of an agent is updated
in each iteration according to the weighted average
of other agents' opinions (according to the \emph{interaction matrix}) and its innate opinion.
The weight that an agent assigns to its own innate opinion
is captured by a {\em resistance parameter} $\alpha_i \in [0,1]$,
where a higher value for the
resistance parameter means that the agent is less
susceptible to persuasion by
the opinions of other agents. 
Under very mild conditions,
the expressed opinions of the agents converge
to an equilibrium, which is a function of the innate opinions, the interaction
matrix between the agents and the agents' resistance parameters.

Recent works by Abebe et al.~\cite{AbebeKPT18}
and Chan et al.~\cite{DBLP:conf/www/ChanLS19}
have considered the opinion dynamics optimization problem,
in which the innate opinions and the interaction matrix between
agents are given as the input, and the goal is to minimize the average equilibrium opinion
by varying the agents' resistance parameter.
As mentioned in their works,
the motivation of the problem has been inspired by
empirical works in development and social psychology that studied people's susceptibility to persuasion,
and more references to related work and applications
are given in~\cite{AbebeKPT18,DBLP:conf/www/ChanLS19}.

Restrictions on how the agents' resistance
parameters may be modified lead to different variants
of the problem with different hardness.
At the trivial end of the spectrum,
if one can choose any $\alpha_i \in [0,1]$ for every agent~$i$,
then the trivial solution to minimize
the average equilibrium is to set $\alpha_i = 1$
for the agent with the minimum innate opinion
and set the resistance of all other agents to 0,
provided that the interaction matrix among
the agents is \emph{irreducible},
in the sense that every agent has some direct or indirect influence over
every other agent.
If the resistance of each agent~$i$
must be chosen from some restricted interval $[l_i, u_i] \subseteq [0,1]$,
then an efficient local search method is given
in~\cite{DBLP:conf/www/ChanLS19} such that
the minimum average equilibrium can be achieved
by setting each agent's resistance parameter to either its lower~$l_i$ or
upper~$u_i$ bound.
In addition to the restriction intervals, the problem gets harder if one places further restrictions
on the number of agents whose resistance parameters may be modified.
The $L_0$-budgeted variant has some initial resistance vector~$\widehat{\alpha}$
for all agents and some budget~$k$,
and the algorithm is allowed to change the resistance parameters of at most~$k$ agents.
Indeed, it is shown in~\cite{AbebeKPT18}
that the $L_0$-budgeted variant is NP-hard via
a reduction from the vertex cover problem.

Intuitively, in the reduction construction
for proving the NP-hardness of the $L_0$-budgeted variant,
the set of agents whose resistance parameters are modified
corresponds to a set of vertices to be considered
as a candidate as a vertex cover in some graph.
Hence, it seems that the binary nature of the choice 
for each agent contributes to the hardness of the problem.
A natural question is whether the problem becomes easier
if one is allowed to make ``fractional'' decision for each agent.
From a practical point of view, the $L_0$-budget uses the implicit
assumption that
modifying the resistance parameter of an agent a little
takes the same effort as modifying it by a lot.  However,
it is reasonable to assume that
the effort it takes to modify an agent's resistance
should be proportional to the magnitude of the change.

With such motivations in mind,
we propose the $L_1$-budgeted variant in this work.
Similar to the $L_0$-budgeted variant in which
an initial resistance vector $\widehat{\alpha}$
and a budget~$k$ is given,
the goal is to minimize the average equilibrium opinion
by choosing a vector~$\alpha$ (satisfying any restriction interval placed
on each agent)
such that $\norm{\alpha - \widehat{\alpha}}_1 \leq k$.

\subsection{Our Contributions}

At first sight, the efficient local search techniques in~\cite{DBLP:conf/www/ChanLS19}
and the fractional nature of the $L_1$-budget
suggest that the problem might be solved optimally by some gradient method or mathematical program.
Indeed, there are examples in which the optimal solution is achieved
by assigning the budget to modify the agents' resistance partially, i.e., the resistance
parameter an agent does not reach its specified lower or upper bound.
However, it turns out that this variant is also NP-hard.

\begin{theorem}
The $L_1$-budgeted variant of the opinion dynamics optimization
problem is NP-hard.
\end{theorem}

\noindent \emph{Hardness Intuition.}
Given an $L_1$-budget,
whether one should spread the budget among many agents or focus it
on a small number of agents depends on the interaction matrix among the agents.
When we tried to understand the structure of the problem by studying various examples,
we discovered that if two agents have little direct or indirect influence on each other,
then the $L_1$-budget should be shared among them.  However,
if all agents are well-connected such as the case of a clique,
then the budget should be as focused as possible on a small number of agents.
Hence, intuitively, the problem should be hard if the underlying interaction matrix
among the agents resembles a well-connected graph.

The difficulty here is that we do not yet know how to quantify the well-connectedness
of the interaction matrix in relation to this \emph{budget-focus} effect.
Furthermore, the hardness still relies on the reduction from the vertex cover problem,
whose hardness has not been extensively studied for graphs with different connectivity.
Our solution to the reduction construction is to consider an interaction matrix
that is a convex combination of a clique and some given graph $G$ that is supposed
to be an instance of the vertex cover problem.

The high level argument is that as long as the weight of the clique is large enough,
the aforementioned budget-focus effect should be in place.  However,
as long as there is a non-zero weight of the graph $G$ on the interaction matrix,
the existence of a vertex cover for $G$ of a certain size will have
a quantifiable effect on the optimal average equilibrium opinion given a certain $L_1$-budget.
Even though the general approach is not too complicated, 
combining these ideas require quite technical calculations.

\noindent \emph{Outline.} In Section~\ref{sec:prelim}, we will introduce the notation and
formally recall various variants of the problem.  In Section~\ref{sec:hardness},
we will give our reduction construction and explain the intuition behind the proofs.
Finally, the most technical details are deferred to Section~\ref{sec:proofs}.

\section{Preliminaries}
	\label{sec:prelim}
	
We recall the problem setting as described in~\cite{AbebeKPT18,DBLP:conf/www/ChanLS19}.
Consider a set $N$ of \emph{agents}, where each agent $i \in N$ is associated with an \emph{innate opinion} $s_i \in [0,1]$ (where higher values correspond to more favorable opinions towards a given topic) and a parameter measuring an agent's susceptibility to persuasion $\alpha_i \in [0,1]$ (where higher values signify agents who are less susceptible to changing their opinions). We call $\alpha_i$ the \emph{resistance parameter}. 

The agents interact with one another in discrete time steps.
%
%
The interaction matrix captures the relationship between agents and is simply
a row stochastic matrix\footnote{
Given sets $U$ and $W$,
we use the notation $U^W$ to denote the collection
of all functions from $W$ to $U$.
Each such function can also be interpreted as a vector (or a matrix if $W$ itself is a
Cartesian product),
where each coordinate is labeled by an element in $W$
and takes a value in $U$.
As an example, a member of $[0,1]^{N \times N}$
is a matrix whose rows and columns
are labeled by elements of $N$.
The alternative notation $[0,1]^{n \times n}$
implicitly assumes a linear ordering on $N$,
which does not have any importance in our case and
would simply be an artefact of the notation.}
 $P \in [0,1]^{N \times N}$
(i.e., each entry of $P$ is non-negative
and every row sums to 1, but $P$ needs not be symmetric).
\ignore{
Equivalently, by defining $A = \Diag(\alpha)$ to be the diagonal  matrix with $A_{ii} = \alpha_i$, $I$ is the identity matrix, and $P \in [0,1]^{V \times V}$ to be the row stochastic matrix (i.e., each entry of $P$ is non-negative
and every row sums to 1) that captures agents interactions, 
}
We denote $A = \Diag(\alpha)$ as the diagonal  matrix with $A_{ii} = \alpha_i$, and $I$ as the identity matrix.
Starting from some arbitrary initial expressed opinion vector $z^{(0)} \in [0,1]^N$,
the expressed opinion vector is updated in each time step according to
the following equation:

\begin{equation}
\label{eq:dynamics2}
z^{(t+1)} := A s + (I - A) P z^{(t)}.
\end{equation}

Equating $z^{(t)}$ with $z^{(t+1)}$, one can see that the equilibrium opinion vector is given by $z = [I - (I-A)P]^{-1} As$, which exists under very mild conditions such as the following.

\begin{fact}[Convergence Assumption]
\label{fact:assume}
 Suppose $P$ is irreducible
and at least one $i \in N$ has $\alpha_i > 0$.
Then, equation~(\ref{eq:dynamics2}) converges to a unique
equilibrium $\lim_{t\rightarrow \infty}z^{(t)}$.
\end{fact}

The opinion susceptibility problem is defined below.  Intuitively,
the objective is to choose a resistance vector $\alpha$ to minimize the sum of equilibrium opinions $\langle \one, z \rangle = \one^\top z$,
i.e., the goal is to drive the average opinion towards~0.
Observe that one can also consider maximizing the sum of equilibrium opinions.
To see that the minimization and maximization problems are equivalent,
consider the transformation $x \mapsto 1 - x$ on the opinion space $[0,1]$
that is applied to the innate opinions and expressed opinions in every time step.
In this paper, we will focus on the minimization problem as follows.

\ignore{
Then, it follows that equation~(\ref{eq:dynamics2})
still holds after the transformation.
Hence, the minimization problem in the original opinion space
is equivalent to the maximization problem in the transformed space.
Either one is an optimization problem in which the goal 
is to drive the average opinion to one of the polar opposites $\{0,1\}$.
Indeed, to get readers to be familiar with both interpretations,
the hardness is proved via the maximization variant
in Section~\ref{sec:hardness},
while the algorithms are given via
the minimization variant in Sections~\ref{sec:efficient}
and~\ref{sec:heuristic}.
}

\begin{definition}[Opinion Susceptibility Problem (Unbudgeted Variant)]
\label{defn:osp}
Given a set $N$ of agents
with innate opinions $s \in [0,1]^N$
and interaction matrix $P \in [0,1]^{N \times N}$,
suppose for each $i \in N$, its resistance is restricted
to some interval $\mcal{I}_i := [l_i, u_i] \subseteq [0,1]$
where we assume\footnote{In view of Fact~\ref{fact:assume},
we assume that for at least one $i$, $l_i > 0$.} that $0 \leq l_i \leq u_i \leq 1$.

The objective is to choose $\alpha \in \mcal{I}_N := \times_{i \in N} \mcal{I}_i \subseteq [0,1]^N$
such that the following objective function is minimized:
$$f(\alpha) := \one^\top [I-(I-A)P]^{-1}As,$$
where $A=\Diag(\alpha)$ is the diagonal matrix with
$A_{ii}=\alpha_i$.  Observe that the assumption in
Fact~\ref{fact:assume} ensures that the above inverse exists.
\end{definition}

\noindent \emph{Budgeted Variants.} 
To describe different types of budgets,
we use the following norms.
Given $x \in \R^N$,
we denote its $L_0$-norm $\norm{x}_0 := |\{i \in N: x_i \neq 0\}|$
and its $L_1$-norm $\norm{x}_1 := \sum_{i \in N} |x_i|$.
For $b \in \{0,1\}$, the $L_b$-budgeted variant
of the problem also has 
some initial resistance vector $\widehat{\alpha} \in \mcal{I}_N$
and a given budget $k > 0$.
The goal is to find $\alpha \in \mcal{I}_N$
to minimize $f(\alpha)$ subject to $\norm{\alpha - \widehat{\alpha}}_b \leq k$.

\noindent \emph{Hardness of the Various Variants.}
A polynomial-time algorithm is given for the unbudgeted variant in~\cite{DBLP:conf/www/ChanLS19},
while the $L_0$-budgeted variant is shown to be NP-hard in~\cite{AbebeKPT18}
via a reduction from vertex cover problem.
The main result of this work is to
show that the $L_1$-budgeted variant is also NP-hard.

\section{Hardness of $L_1$-Budgeted Variant}
	\label{sec:hardness}

	As in~\cite{AbebeKPT18}, we shall prove the NP-hardness
	of the $L_1$-budgeted variant via reduction from the
	vertex cover problem on regular graphs~\cite{Feige03},
	where a graph is $d$-regular if every vertex has degree~$d$.
	
	\begin{fact}[Vertex Cover on Regular Graphs]
	Given a $d$-regular undirected graph $G = (V,E)$ and some $k > 0$,
	it is NP-hard (even for $d=3$) to decide if $G$ has a vertex cover $T$ of size $k$,
	where $T  \subseteq V$ is a vertex cover for $G$ if every edge in $E$ has
	at least one end-point in $T$.
	\end{fact}

\subsection{Warmup: Reduction for $L_0$-Budget}
\label{sec:L_0}

Before we give our final reduction construction for $L_1$-budget,
we give a simplified reduction for $L_0$-budget, which will
offer some intuition on why the $L_1$-budget reduction is more complicated.  
The reduction here for $L_0$-budget
is similar to the one given in~\cite{AbebeKPT18},
but is even simpler because we allow different agents~$i$ to have different
ranges $[l_i, u_i]$ for their resistance parameters.

Recall that
an instance of the vertex cover problem
consists of a $d$-regular graph $G=(V,E)$
with $n = |V|$
and some target vertex cover size $k$.

\noindent \emph{Reduction Construction.} In addition to the original
vertices in $V$, we create one extra agent~0 to form $N := V \cup \{0\}$.
For the innate opinions, $s_0 = 1$ and $s_i = 0$ for $i \neq 0$;
for the initial resistances,
$\widehat{\alpha}_0 = 1$ and $\widehat{\alpha}_i = 0$ for $i \neq 0$.
For the range of resistance parameter,
we restrict $\mcal{I}_0 = \{1\}$
and $\mcal{I}_i = [0,1]$ for $i \neq 0$;
in other words, agent~$0$ will always remain the most stubborn,
while agents in $V$ have 0 resistance initially,
but their resistance parameters could be increased to 1
subject to the budget constraint.
The $L_0$-budget is $k$, which is the same as the target cover size.
Our final construction for the $L_1$-budget will also share the above parameter settings.

\noindent \emph{Interaction Matrix for $L_0$-Budget Reduction.} We
next describe the interaction matrix $P$.
For agent~0, we will always have resistance $\alpha_0 = 1$,
and so the corresponding row in $P$ is irrelevant, but we
could set $P_{0i} = \frac{1}{n}$ for $i \in V$ to be concrete.
For $i \neq 0$, let $P_{ij} = \frac{1}{d+1}$
if $j = 0$ or $\{i,j\} \in E$ is an edge in $G$,
recalling that each node~$i \in V$ has degree~$d$ in $G$.
As we shall see, the
reduction for $L_1$-budget will have a different interaction matrix.

\noindent \emph{Intuition.}  For the $L_0$-budgeted variant,
if we wish to change the resistance parameter of an agent~$i \in V$
(who has innate $s_i = 0$),
we might as well set it to $\alpha_i = 1$,
because the goal is to minimize the expressed opinion.
To complete the reduction proof, it suffices
to give a threshold~$\vartheta$ for the objective
function~$f$ that can distinguish between the YES and NO instances
of the vertex cover problem.
The following two lemmas complete the reduction argument.

\begin{lemma}[YES Instance]
\label{lemma:yes}
Let $\vartheta := 1 + \frac{n-k}{d+1}$.
Suppose $G = (V,E)$ has a vertex cover~$T$ of size~$k$.
Then, by changing the resistance parameters to $\alpha_i = 1$
for all~$i \in T$ (while those for other agents are not changed),
we can achieve $f(\alpha) = \vartheta$.
\end{lemma}

\begin{proof}
We compute the equilibrium expressed opinion of each agent.
For agent~$0$, we have $z_0 = 1$;
for $i \in T$ in the vertex cover, we have $z_i = 0$.

For $i \in V \setminus T$,
we still have $\alpha_i = 0$. Since
all its neighbors in $V$ are in $T$
and~$i$ is influenced by agent~$0$,
we have $z_i = \frac{1}{d+1}$.

Therefore, $f(\alpha) = \sum_{i \in N} z_i =  1 + |T| \cdot 0 + |V \setminus T| \cdot \frac{1}{d+1} = \vartheta$,
as required.
\end{proof}

\begin{lemma}[NO Instance]
\label{lemma:no}
Suppose $G = (V,E)$ has no vertex cover of size~$k$.
Then, for any $\alpha \in \mcal{I}_N$
such that $\norm{\alpha - \widehat{\alpha}}_0 \leq k$,
$f(\alpha) \geq \vartheta + \frac{2}{d(d+1)}$.
\end{lemma}

\begin{proof}
Since the goal is to minimize~$f$,
the minimum can be achieved
by using all of the $L_0$-budget~$k$.
Moreover, as each~$i \in V$ has innate $s_i = 0$,
if we change its resistance parameter,
we should set it to $\alpha_i = 1$.
Hence, we can assume 
that there is some $T \subseteq V$ of size $|T| = k$
such that $\alpha_i = 1$ for $i \in T$
and $\alpha_i = 0$ for $i \in V \setminus T$.
We remark that to reach the same
conclusion for the $L_1$-budget reduction later will 
require a lot more technical details.

As in Lemma~\ref{lemma:yes},
we can conclude $z_0 = 1$
and $z_i = 0$ for $i \in T$.

For $i \in V \setminus T$,
we now only have the inequality
$z_i \geq \frac{1}{d+1}$.
However, we can now achieve a stronger
lower bound because $T$ is not a vertex cover for $G$.

Let $\gamma$ be the minimum
over all $z_i$ such that $i$ is incident on
an edge in $E$ that is not covered by $T$.
Suppose $i \in V \setminus T$ is 
a vertex that attains $\gamma$
and the edge $\{i, j\}$ is not covered by $T$.
Then, we have $z_j \geq \gamma$.  Hence,
we have $\gamma = z_i \geq \frac{1 + \gamma}{d+1}$,
which implies that $\gamma \geq \frac{1}{d} = \frac{1}{d+1} + \frac{1}{d(d+1)}$.

Since there is at least one edge in $E$ that is not covered by $T$,
we have
$f(\alpha) \geq \vartheta + \frac{2}{d(d+1)}$,
as required.
\end{proof}

\subsection{Reduction for $L_1$-Budget}
\label{sec:L_1}

We first describe the main challenge for adapting
the reduction proof for $L_0$-budget
to $L_1$-budget.  The issue is that to adapt 
Lemma~\ref{lemma:no} for the NO instances 
of the vertex cover problem,
we need a desirable structural property on an optimal
solution for the $L_1$-budgeted variant of the opinion optimization problem.
Specifically, we would like to argue
that to minimize the objective function~$f$ with an integral budget~$k$,
one should pick a subset $T \subseteq V$
of exactly~$k$ agents on whom to use the budget,
as opposed to spreading the budget fractionally over more than~$k$ agents.

Unfortunately, this is not true for the reduction
construction given in Section~\ref{sec:L_0}.
Indeed, we have discovered that for two agents $i$ and $j$ that
are somehow not ``well-connected'' in $G$,
if some fixed $L_1$-budget of less than 2 
is assigned to them,
spreading  the budget fractionally among the two agents
would yield a lower objective value
than biasing the budget towards one agent.
On the other hand, we discovered that the desirable structural property holds
in some cases where all the vertices in $G$ are ``well-connected'',
for instance, if $G$ is a clique on the $n$ vertices.
However, since there is no connectivity assumption 
on the given instance $G$ of vertex cover,
we consider the following interaction matrix.

\noindent \emph{Interaction Matrix for $L_1$-Budget Reduction.}
Recall that we are given a $d$-regular graph $G=(V,E)$ with $n = |V|$,
and $N = V \cup \{0\}$.
Let $C \in [0,1]^{N \times N}$ be a row-stochastic matrix
such that for $i \neq j$, $C_{ij} = \frac{1}{n}$,
recalling that $|N| = n + 1$; in other words, $C$ behaves like a clique on $N$.
Let $R \in [0,1]^{N \times N}$ be a row-stochastic matrix
such that $R_{00} = 1$,
and $R_{ij} = \frac{1}{d}$ \emph{iff} $\{i, j\} \in E$,
where all other entries of $R$ are 0;
in other words, $R$ is the normalized adjacency matrix of $G$ with an additional isolated vertex~$0$.
For some appropriate $\delta \in (0,1)$ (that depends only on $d$ and $n$),
we consider the following
interaction matrix $P^{(\delta)} := (1 - \delta) C + \delta R$.

Recall that the innate opinions
and initial resistance parameters
are the same as in Section~\ref{sec:L_0},
i.e., $s_0 = \widehat{\alpha}_0 = 1$
and $s_i = \widehat{\alpha}_i = 0$ for $i \in V$.
Moreover, the $L_1$-budget can be used
to change the resistance parameters of only agents in~$V$,
i.e., $\alpha_0 = 1$ must remain.

The lemma for YES instance is similar as that in Section~\ref{sec:L_0}
We define the threshold
$\vartheta := 1 + \frac{(1 - \delta)(n-k)}{n - (1 - \delta)(n-k-1)}$.

\begin{lemma}[YES Instance]
\label{lemma:yes2}
Suppose $G = (V,E)$ has a vertex cover~$T$ of size~$k$.
Then, by changing the resistance parameters to $\alpha_i = 1$
for all~$i \in T$ (while those for other agents are not changed),
we can achieve $f(\alpha) = \vartheta$.
\end{lemma}

\begin{proof}
As in Lemma~\ref{lemma:yes},
the equilibrium expressed opinions
are $z_0 = 1$ and $z_i = 0$ for $i \in T$.

For $j \in V \setminus T$, recall that $\alpha_j = 0$ and
we exploit the symmetry in $P_{\delta} = (1 - \delta) C + \delta R$
to analyze the value of $z_j$.
Note that with respect to $C$, agent~$j$ observes that agent~$0$
has $z_0 = 1$ and the $k$ agents $i \in T$ have $z_i = 0$,
and there are $n-k-1$ other agents like itself;
with respect to $R$, agent~$j$ observes that all its $d$ neighbors in $G$ are
in $T$, and so have $z_i = 0$.

Therefore,  every agent~$j \in V \setminus T$
has this same observation, and we can conclude that $z_j$'s
have some common value $\gamma$ for all $j \in V \setminus T$ satisfying:

$\gamma = (1 - \delta) \cdot \frac{1}{n} \cdot (1 + k \cdot 0 + (n-k-1) \cdot \gamma).$

This gives $\gamma = \frac{1 - \delta}{n - (1 - \delta)(n-k-1)}$,
and so $f(\alpha) = 1 + k \cdot 0 + (n-k) \cdot \gamma = \vartheta$.
\end{proof}

The following structural property is needed
for the analysis of NO instances.
Its proof is technical and is deferred to Section~\ref{sec:proofs}.

\begin{lemma}[Structural Property of Optimal Solution]
\label{lemma:structure}
Suppose $\delta := \frac{d^3(2d-1)^{3n-3}}{(n+1)^6(2d+1)^{3n}}$ is set to define the
interaction matrix $P^{(\delta)}$ above,
and an $L_1$-budget of $k$ is given
to change the resistance parameters of agents in~$V$.
Then, the objective function $f$ can be minimized
by picking some $T \subseteq V$ of size $k$,
and setting $\alpha_i = 1$ for $i \in T$.
\end{lemma}

\begin{lemma}[NO Instance]
\label{lemma:no2}
Suppose $G = (V,E)$ has no vertex cover of size~$k$,
and $\delta \in (0,1)$ is chosen
to satisfy Lemma~\ref{lemma:structure}.
Then, for any $\alpha \in \mcal{I}_N$
such that $\norm{\alpha - \widehat{\alpha}}_1 \leq k$,
$f(\alpha) \geq \vartheta + 
\frac{\delta}{dn}$.
\end{lemma}

\begin{proof}
Because of Lemma~\ref{lemma:structure},
we can assume that
the minimum $f(\alpha)$ is achieved
by picking some $T \subseteq V$
of size~$k$
and set $\alpha_i = 1$ for $i \in T$
and $\alpha_j = 0$ remains for $j \in V \setminus T$.

Again, the equilibrium expressed opinions
satisfy $z_0 = 1$ and $z_i = 0$ for $i \in T$.

Let $\gamma := \min_{j \in V \setminus T} z_j$.
Then, a similar argument as in Lemma~\ref{lemma:yes2}
gives the inequality $\gamma \geq \gamma_0 := \frac{1 - \delta}{n - (1 - \delta)(n-k-1)}$.

We can get a stronger lower bound
because $T$ is not a vertex cover for $G = (V,E)$.
Let $\widehat{\gamma}$
be the minimum $z_j$ among $j \in V \setminus T$
such that $j$ is an end-point of an edge not covered by $T$.

Then, we have the inequality
$\widehat{\gamma} \geq (1 - \delta) \cdot \frac{1}{n} \cdot (1 + k \cdot 0 + (n-k-1) \cdot \gamma_0)
+ \delta \cdot \frac{\widehat{\gamma}}{d} = \gamma_0 + \widehat{\gamma} \cdot \frac{\delta}{d}$.

Hence, we have $\widehat{\gamma} \geq (1 - \frac{\delta}{d})^{-1} \cdot \gamma_0 \geq \gamma_0 + 
\frac{\delta \gamma_0}{d}$.

Since there is at least one edge in $E$ that is not covered by $T$ (and such an edge has two end-points),
we have
$f(\alpha) \geq \vartheta + \frac{2\delta \gamma_0}{d} \geq \vartheta + \frac{\delta}{dn}$, as required.
\end{proof}

\begin{corollary}
It is NP-hard to solve the $L_1$-budgeted variant with
additive error at most $\frac{\delta}{dn}$ on the objective function $f$.
\end{corollary}

\begin{proof}
Observe that from Lemmas~\ref{lemma:yes2} and~\ref{lemma:no2},
it suffices to use a precision of $\Theta(\frac{\delta}{dn})$ on the objective function,
which can be achieved using $O(\log \frac{dn}{\delta}) = \poly(n)$ number of bits.
\end{proof}

\section{Technical Proofs for Well-Connected Interaction Matrices}
\label{sec:proofs}

To complete our hardness proof in Section~\ref{sec:L_1},
it suffices to proof Lemma~\ref{lemma:structure}.
Recall that the goal is to pick some $\delta \in (0,1)$
to define the interaction matrix
$P^{(\delta)} = (1 - \delta) C + \delta R$
such that given an $L_1$-budget~$k \in \Z$,
the objective function $f(\alpha)$
can be minimized by allocating the budget to
exactly~$k$ agents in~$V$.
One way to achieve this is to pick two arbitrary agents~$i,j \in V$
and show that if some fixed $L_1$-budget~$b < 2$
is assigned for~$i$ and $j$ (while the other agents are not modified),
then the objective function $f(\alpha)$
can be minimized by prioritizing the budget to either~$i$ or $j$
as much as possible.
Denoting $e_i \in [0,1]^N$ as the unit vector with~$i \in N$
as the only non-zero coordinate,
we shall prove the following.

\noindent \textbf{Formal Goal.} By setting $\delta = \frac{d^3(2d-1)^{3n-3}}{(n+1)^6(2d+1)^{3n}}$,
we will show that if~$i,j \in V$ are such that $0 \leq \alpha_i, \alpha_j < 1$,
then the second derivative of~$f$ in the direction $e_i - e_j$ satisfies:

\begin{equation}
\label{eq:Hessian}
(e_i - e_j)^\top  \nabla f(\alpha) = 0 \implies (e_i - e_j)^\top  \nabla^2 f(\alpha) (e_i - e_j) < 0.
\end{equation}

Statement~(\ref{eq:Hessian})
implies that if two agents in~$V$ both receive
non-zero fractional budget to change their resistance parameters,
then it will not increase the objective function~$f$
by biasing the budget towards one of them.

\noindent \textbf{Notation Recap.} Recall that given $\alpha \in [0,1]^N$,
we write $A := \Diag(\alpha)$.
Moreover, we denote $X=X(\alpha) := I - (I-A)P^{(\delta)}$.
Under conditions such as Fact~\ref{fact:assume},
we write $M = M(\alpha) := X^{-1}$ and the equilibrium vector $z = z(\alpha) := M A s$,
where $s \in [0,1]^N$ is the innate vector from Section~\ref{sec:L_1} such that $s_0 = 1$ and $s_i = 0$ for $i \in V$.
Finally, the objective function is $f(\alpha) := \one^\top  z(\alpha)$.
Note that the quantities have a dependence on~$\delta$,
and we will use a superscript such as~$f^{(\delta)}$ when we wish
to emphasize this dependence.

\begin{fact}[Technical Calculations~\cite{DBLP:conf/www/ChanLS19}]
\label{fact:tech}
Whenever the above quantities are well-defined, we have
\begin{itemize}

\item $M \geq 0$ and $M_{ii} \geq 1$ for all $i \in N$.

\item If $\alpha_k \neq 1$, $(PM)_{kk} = \frac{M_{kk}-1}{1 - \alpha_k}$
and $(PM)_{kj} = \frac{M_{kj}-1}{1 - \alpha_k}$
for $j \neq k$.

\item For $i \in N$, $e_i^\top  \nabla f (\alpha) = \frac{\partial f}{\partial \alpha_i}
= \frac{s_i - z_i(\alpha)}{1 - \alpha_i} \cdot \one^\top  M e_i$.

\end{itemize}
\end{fact}

Statement~(\ref{eq:Hessian}) inspires us to analyze the following quantities.

	\begin{lemma}
		\label{dirderi}
		Suppose $i,j \in V$ such that $0 \leq \alpha_i, \alpha_j < 1$
		and $(e_i - e_j)^\top  \nabla f(\alpha) = 0$.
		Then, we have

		\begin{equation*}
			(e_i-e_j)^\top  \nabla^2 f(\alpha)(e_i-e_j)
			=\frac{2}{(1-\alpha_i)(1-\alpha_j)}\cdot
			\left\{z_i(\alpha)y_{ij}(\alpha)+z_j(\alpha)y_{ji}(\alpha)\right\}.
		\end{equation*}
		where $y_{ij}(\alpha):=\mathbf{1}^\top  Me_i\cdot(M_{jj}-1)-\mathbf{1}^\top  Me_j\cdot M_{ji}$.
		
	\end{lemma}
	
	\begin{proof}
		By the definition of the inverse of a matrix $B$, we have $BB^{-1}=I$. The partial derivative with respect to a variable $t$ is: $\frac{\partial B}{\partial t}B^{-1}+B \frac{\partial B^{-1}}{\partial t}=0$. Hence, we have $\frac{\partial B^{-1}}{\partial t}=-B^{-1} \frac{\partial B}{\partial t}B^{-1}$. Applying the above result with $B=I-(I-A)P$ and $t=\alpha_i$, we get $\frac{\partial M}{\partial \alpha_i}=-Me_ie_i^\top PM$. Combining with the results in Fact~\ref{fact:tech}, we obtain the second-order partial derivatives of $f$ as follows:
		\begin{align*}
			\frac{\partial^2 f(\alpha)}{\partial \alpha^2_i} &=\frac{s_i-z_i(\alpha)}{1-\alpha_i}\mathbf{1}^\top \frac{\partial M}{\partial \alpha_i} e_i+\frac{s_i-z_i(\alpha)}{(1-\alpha_i)^2}\mathbf{1}^\top  Me_i-\frac{1}{1-\alpha_i}\mathbf{1}^\top  Me_i\frac{\partial z_i(\alpha)}{\partial \alpha_i}\\
			&=-\frac{s_i-z_i(\alpha)}{1-\alpha_i}\mathbf{1}^\top  Me_i e^\top _i PMe_i+\frac{s_i-z_i(\alpha)}{(1-\alpha_i)^2}\mathbf{1}^\top  Me_i-\frac{s_i-z_i(\alpha)}{(1-\alpha_i)^2}\mathbf{1}^\top  Me_i e^\top _i Me_i\\
			&=-\frac{s_i-z_i(\alpha)}{(1-\alpha_i)^2}\mathbf{1}^\top  Me_i\cdot(M_{ii}-1)+\frac{s_i-z_i(\alpha)}{(1-\alpha_i)^2}\mathbf{1}^\top  Me_i-\frac{s_i-z_i(\alpha)}{(1-\alpha_i)^2}\mathbf{1}^\top  Me_i \cdot M_{ii}\\
			&=-\frac{s_i-z_i(\alpha)}{(1-\alpha_i)^2}\left[\mathbf{1}^\top  Me_i\cdot(M_{ii}-1)-\mathbf{1}^\top  Me_i+\mathbf{1}^\top  Me_i \cdot M_{ii}\right]\\
			&=-2\cdot \frac{s_i-z_i(\alpha)}{(1-\alpha_i)^2}\mathbf{1}^\top  Me_i\cdot(M_{ii}-1),\\
			\frac{\partial^2 f(\alpha)}{\partial \alpha_i \partial \alpha_j} &=\frac{s_i-z_i(\alpha)}{1-\alpha_i}\mathbf{1}^\top \frac{\partial M}{\partial \alpha_j} e_i-\frac{1}{1-\alpha_i}\mathbf{1}^\top  Me_i\frac{\partial z_i(\alpha)}{\partial \alpha_j}\\
			&=-\frac{s_i-z_i(\alpha)}{1-\alpha_i}\mathbf{1}^\top  Me_j e^\top _j PMe_i-\frac{s_j-z_j(\alpha)}{(1-\alpha_i)(1-\alpha_j)}\mathbf{1}^\top  Me_i e^\top _i Me_j\\
			&=-\frac{s_i-z_i(\alpha)}{1-\alpha_i}\mathbf{1}^\top  Me_j\cdot\frac{M_{ji}}{1-\alpha_j}-\frac{s_j-z_j(\alpha)}{(1-\alpha_i)(1-\alpha_j)}\mathbf{1}^\top  Me_i\cdot M_{ij}\\
			&=-\frac{s_i-z_i(\alpha)}{(1-\alpha_i)(1-\alpha_j)}\mathbf{1}^\top  Me_j \cdot M_{ji} -\frac{s_j-z_j(\alpha)}{(1-\alpha_i)(1-\alpha_j)}\mathbf{1}^\top  Me_i\cdot M_{ij}.
		\end{align*}
				
		Therefore, we have:
		
		\begin{align*}
			&~~~~(e_i-e_j)^\top  \nabla^2 f(\alpha)(e_i-e_j)\\
			&=\frac{\partial^2 f(\alpha)}{\partial \alpha^2_i}
			+\frac{\partial^2 f(\alpha)}{\partial \alpha^2_j}
			-\frac{\partial^2 f(\alpha)}{\partial \alpha_i \partial \alpha_j}
			-\frac{\partial^2 f(\alpha)}{\partial \alpha_j \partial \alpha_i}\\
			&=2\cdot \frac{z_i(\alpha)}{(1-\alpha_i)^2}\mathbf{1}^\top  Me_i\cdot(M_{ii}-1)+2\cdot \frac{z_j(\alpha)}{(1-\alpha_j)^2}\mathbf{1}^\top  Me_j\cdot(M_{jj}-1)-\\
			&~~~~2\cdot \frac{z_i(\alpha)}{(1-\alpha_i)(1-\alpha_j)}\mathbf{1}^\top  Me_j\cdot M_{ji}-2\cdot \frac{z_j(\alpha)}{(1-\alpha_i)(1-\alpha_j)}\mathbf{1}^\top  Me_i \cdot M_{ij}.
		\end{align*}
		
		From the hypothesis of the lemma, we have:
		
		\begin{equation*}
			(e_i-e_j)^\top \nabla f(\alpha)
			=\frac{\partial f(\alpha)}{\partial \alpha_i}-\frac{\partial f(\alpha)}{\partial \alpha_j}
			=-\frac{z_i(\alpha)}{1-\alpha_i}\mathbf{1}^\top  Me_i+\frac{z_j(\alpha)}{1-\alpha_j}\mathbf{1}^\top  Me_j=0,
		\end{equation*}
		which gives
		\begin{equation*}
			\frac{z_i(\alpha)}{1-\alpha_i}\mathbf{1}^\top  Me_i=\frac{z_j(\alpha)}{1-\alpha_j}\mathbf{1}^\top  Me_j.
		\end{equation*}    	
		Thus, we have
		\begin{align*}
			&~~~~(e_i-e_j)^\top  \nabla^2 f(\alpha)(e_i-e_j)\\
			&=2\cdot \frac{z_j(\alpha)}{(1-\alpha_i)(1-\alpha_j)}\mathbf{1}^\top  Me_j\cdot(M_{ii}-1)+2\cdot \frac{z_i(\alpha)}{(1-\alpha_i)(1-\alpha_j)}\mathbf{1}^\top  Me_i\cdot(M_{jj}-1)-\\
			&~~~~2\cdot \frac{z_i(\alpha)}{(1-\alpha_i)(1-\alpha_j)}\mathbf{1}^\top  Me_j \cdot M_{ji}-2\cdot \frac{z_j(\alpha)}{(1-\alpha_i)(1-\alpha_j)}\mathbf{1}^\top  Me_i \cdot M_{ij}\\
			&=\frac{2}{(1-\alpha_i)(1-\alpha_j)}\cdot \left\{z_i(\alpha)\left[\mathbf{1}^\top  Me_i\cdot(M_{jj}-1)-\mathbf{1}^\top  Me_j \cdot M_{ji}\right]+\right.\\
			&~~~~~~~~~~~~~~~~~~~~~~~~~~~~~~~~~~~~~\left.z_j(\alpha)\left[\mathbf{1}^\top  Me_j\cdot (M_{ii}-1)-\mathbf{1}^\top  Me_i\cdot M_{ij}\right]\right\}\\
			&=\frac{2}{(1-\alpha_i)(1-\alpha_j)}\cdot
			\left[z_i(\alpha)y_{ij}(\alpha)+z_j(\alpha)y_{ji}(\alpha)\right].
		\end{align*}
		as desired.
	\end{proof}

In view of Lemma~\ref{dirderi}, it suffices to analyze the quantity $y_{ij}$.
Since every~$i \in V$ with $\alpha_i<1$ is influenced by agent~$0$ (with $z_0 = 1$) in $P^{(\delta)}$,
we have $z_i(\alpha) > 0$ for all $i \in V$.
Hence, to show statement~(\ref{eq:Hessian}),
it remains to show that $y_{ij}(\alpha) < 0$.
We next analyze $y_{ij}$ as functions of
$\alpha_j$ and $P$ respectively in the next two lemmas.

%
%
	
	\begin{lemma}[Monotonicity]
		\label{monotone}
		Given $\alpha \in [0, 1]^{n}$, let $A:=\Diag(\alpha)$ and $P$ be a row-stochastic matrix such that $M:=[I-(I-A)P]^{-1}$ exists.
		For $i\ne j \in V$, we fix $\alpha_k$ for $k \neq j$ and all entries of $P$ and consider the following quantity
		as a function of $\alpha_j$:
		
		$$y_{ij}(\alpha_j):=\mathbf{1}^\top  Me_i\cdot(M_{jj}-1)-\mathbf{1}^\top  Me_j\cdot M_{ji.}$$
		
		Then, $y_{ij}(\alpha_j)$ is a strictly monotone or constant function of $\alpha_j$ on $[0,1]$, i.e., it is either strictly increasing, strictly decreasing, or constant. In addition, $y_{ij}(1)=0$.
	\end{lemma}
	
	\begin{proof}
		When $\alpha_j=1$, note that the $j$-th row of the matrix $I-(I-A)P$ is equal to $e_j^\top $. By considering the $j$-th row of the  equation $[I-(I-A)P]M=I$, we have $M_{jj}=1$ and $M_{jk}=0$ for any $k\ne j$. Hence, $M_{jj}-1=M_{ji}=0$ and so $y_{ij}(1)=0$.
		
		We will now show that $y_{ij}(\alpha_j)$ is a strictly monotone function of $\alpha_j$. Notice that $y_{ij}$ is a continuous function of $\alpha_j$ since $y_{ij}$ is a continuous function of $M$, $M$ is a continuous function of $\alpha_j$ (because of the continuity of matrix inversion), and a composition of continuous functions is continuous. In addition, we know that $\frac{\partial B^{-1}}{\partial t}=-B^{-1}\frac{\partial B}{\partial t}B^{-1}$ for any invertible matrix $B$. Applying the above result with $B=I-(I-A)P$ and $t=\alpha_i$, we get
		\begin{equation*}
			\frac{\partial M}{\partial \alpha_i}
			=-Me_ie_i^\top PM.
		\end{equation*}
		Hence, when $\alpha_j\ne 1$, the partial derivative of $y_{ij}$ with respect to $\alpha_j$ is
		\begin{align*}
			\frac {\partial y_{ij}} {\partial \alpha_j} 
			&= \frac {\partial} {\partial \alpha_j} \left[1^\top Me_i\cdot(e_j^\top  M e_j -1)-1^\top Me_j\cdot e_j^\top  M e_i \right] \\
			&=1^\top  \frac{\partial M}{\partial \alpha_j} e_i\cdot (e_j^\top  M e_j-1) + 1^\top  M e_i\cdot e_j^\top  \frac{\partial M}{\partial \alpha_j} e_j -1^\top  \frac{\partial M}{\partial \alpha_j} e_j\cdot e_j^\top  M e_i - 1^\top  M e_j\cdot e_j^\top  \frac{\partial M}{\partial \alpha_j} e_i\\
			&=-1^\top  Me_je_j^\top PM e_i\cdot(e_j^\top  M e_j -1) - 1^\top  M e_i\cdot e_j^\top  Me_je_j^\top PM e_j+\\
			&~~~~~1^\top  Me_je_j^\top PM e_j\cdot e_j^\top  M e_i + 1^\top  M e_j\cdot e_j^\top  Me_je_j^\top PM e_i\\
			&=-1^\top  Me_j\cdot \frac{M_{ji}}{1-\alpha_j}\cdot(M_{jj} -1) - 1^\top  M e_i\cdot M_{jj} \cdot \frac{M_{jj}-1}{1-\alpha_j}+\\
			&~~~~~1^\top  Me_j\cdot \frac{M_{jj}-1}{1-\alpha_j}\cdot M_{ji} + 1^\top  M e_j\cdot M_{jj}\cdot \frac{M_{ji}}{1-\alpha_j} & \text{(Fact~\ref{fact:tech})} \\
			&=1^\top  M e_j\cdot M_{jj}\cdot \frac{M_{ji}}{1-\alpha_j} - 1^\top  M e_i\cdot M_{jj} \cdot \frac{M_{jj}-1}{1-\alpha_j}\\
			&=-\frac{M_{jj}}{1-\alpha_j}[1^\top  M e_i \cdot (M_{jj}-1) - 1^\top  M e_j \cdot M_{ji}]\\
			&=-\frac{M_{jj}}{1-\alpha_j} y_{ij},
		\end{align*}
		where $M_{jj}\geq 1 > 0$ by Fact~\ref{fact:tech}. 
		
		Rewriting $g(t) = - y_{ij}(1-t)$, 
		we have an alternative form $\frac{dg}{dt} = \frac{M(t) \cdot g}{t}$,
		where $M(t) \geq 1$ and $g(0) = y_{ij}(1) = 0$.  It follows that if $g$ is a continuous
		function, then either $g$ stays 0 in [0,1] or $g$ is strictly monotone.
		\end{proof}

Recall that our goal
is to choose some $\delta > 0$
to define the interaction matrix $P^{(\delta)} := (1 - \delta) C + \delta R$
such that we can prove that the quantity $y_{ij}^{(\delta)} < 0$.
The next lemma shows that for the special case $\delta = 0$,
we can argue that $y_{ij}^{(0)} < 0$ for $\alpha_j < 1$.

\begin{lemma}
		\label{lemma:claimC}
For $\delta = 0$, consider $P = P^{(0)} = C$, whose diagonal entries are 0 and every other entry is $\frac{1}{n}$.
For $\alpha \in [0,1]^N$ (with $\alpha_0 =1$), we have $A:=\Diag(\alpha)$ and
$M=[I-(I-A)C]^{-1}$.  Fix some $i \neq j \in V$, and consider			
		\begin{equation*}
			y_{ij}(\alpha_j):=\mathbf{1}^\top  Me_i\cdot(M_{jj}-1)-\mathbf{1}^\top  Me_j\cdot M_{ji}.
		\end{equation*}
		as a function of $\alpha_j$; when $\alpha_j=0$, $y_{ij}(0) \leq -\frac{1}{n+1}$.  Moreover,
		$\mathbf{1}^\top M\mathbf{1} \leq n$.
	\end{lemma}
	
	\begin{proof}
	We write $C := \frac{1}{n}(J-I)$, where every entry in $J$ is 1.
		By the Sherman-Morrison formula,
		\begin{align*}
			C^{-1} &=\left[\frac{1}{n}(J-I)\right]^{-1}=n\left[-I+\mathbf{1}\mathbf{1}^\top \right]^{-1} =n\left[-I-\frac{(-I)\mathbf{1}\mathbf{1}^\top (-I)}{1+\mathbf{1}^\top (-I)\mathbf{1}}\right]\\
			&=n\left[-I-\frac{J}{1-(n+1)}\right] =J-nI,
		\end{align*}
		so we have
		\begin{equation*}
			M=\left[I-(I-A) C\right]^{-1}=\left[C^{-1} C-(I-A) C\right]^{-1}=C^{-1}\left(C^{-1}-I+A\right)^{-1}
			=(J-nI)[J+A-(n+1)I]^{-1}.
		\end{equation*}
		Denote $D:=A-(n+1)I$. Then $D$ is a diagonal matrix with $D_{ii}=\alpha_i-n-1<0$  for each $i$ and
		\begin{equation*}
			M=(J-nI)(J+D)^{-1}.
		\end{equation*}
		By the Sherman-Morrison formula again, we have
		\begin{equation*}
			(J+D)^{-1}=(D+\mathbf{1}\mathbf{1}^\top )^{-1}=D^{-1}-\frac{D^{-1}\mathbf{1}\mathbf{1}^\top  D^{-1}}{1+\mathbf{1}^\top D^{-1}\mathbf{1}}=D^{-1}-\frac{D^{-1}J D^{-1}}{1+\displaystyle\sum_{k \in N}\frac{1}{\alpha_k-n-1}}.
		\end{equation*}
		Since $D^{-1}$ exists, $M^{-1}$ also exists.
		
		Denote $w:=\displaystyle\sum_{k \in N}\frac{1}{\alpha_k-n-1}<0$. Then
		$(J+D)^{-1}=D^{-1}-\displaystyle \frac{D^{-1}J D^{-1}}{1+w}$.
		We then have
		\begin{equation*}
			M=(J-nI)\left[D^{-1}-\frac{D^{-1}J D^{-1}}{1+w}\right]
			=JD^{-1}-nD^{-1}-\frac{JD^{-1}JD^{-1}}{1+w}+\frac{nD^{-1}JD^{-1}}{1+w}.
		\end{equation*}
		Since $JD^{-1}J=\mathbf{1}\mathbf{1}^\top D^{-1}\mathbf{1}\mathbf{1}^\top 
		=\mathbf{1}\displaystyle\left(\sum_{k \in N}\frac{1}{\alpha_k-n-1}\right)\mathbf{1}^\top =wJ$,
		\begin{align}
			M&
			=JD^{-1}-nD^{-1}-\frac{w}{1+w}JD^{-1}+\frac{n}{1+w}D^{-1}JD^{-1} \nonumber\\
			&=\frac{1}{1+w}JD^{-1}-nD^{-1}+\frac{n}{1+w}D^{-1}JD^{-1} \label{M-expression-P=C}
		\end{align}
		and hence
		\begin{align*}
			\mathbf{1}^\top M\mathbf{1}
			&=\frac{1}{1+w}\mathbf{1}^\top JD^{-1}\mathbf{1}-n\mathbf{1}^\top D^{-1}\mathbf{1}+\frac{n}{1+w}\mathbf{1}^\top D^{-1}JD^{-1}\mathbf{1} \\
			&=\frac{1}{1+w}\mathbf{1}^\top \mathbf{1}\mathbf{1}^\top D^{-1}\mathbf{1}-n\mathbf{1}^\top D^{-1}\mathbf{1}+\frac{n}{1+w}(\mathbf{1}^\top D^{-1}\mathbf{1})^2\\
			&=\frac{(n+1)w}{1+w}-nw+\frac{nw^2}{1+w}
			=\frac{w}{1+w}.
		\end{align*}
		Note that
		\begin{equation*}
			1+w=1+\displaystyle\sum_{k \in N}\frac{1}{\alpha_k-n-1}
			=\displaystyle\sum_{k \in N}\left(\frac{1}{n+1}+\frac{1}{\alpha_k-n-1}\right),
		\end{equation*}
		and
		\begin{equation*}
			\frac{1}{n+1}+\frac{1}{\alpha_k-n-1}
			=\frac{\alpha_k-n-1+n+1}{(n+1)(\alpha_k-n-1)}
			=\frac{\alpha_k}{(n+1)(\alpha_k-n-1)}
			\leq 0
		\end{equation*}
		for any $k\in N$ with strict inequality holds for $k=0$ since $\alpha_0=1$.
		
		Therefore, $1+w<0$ and thus
		\begin{equation*}
			\mathbf{1}^\top M\mathbf{1}
			=\frac{w}{1+w}=\frac{\sum_{k \in N}\frac{1}{\alpha_k-n-1}}{1+\sum_{k \in N}\frac{1}{\alpha_k-n-1}}
			\leq\frac{\frac{n+1}{-(n+1)}}{1+\frac{n+1}{-n}}
			=\frac{1}{-1+\frac{n+1}{n}}
			=n.
		\end{equation*}
		
		Also, from Equation \ref{M-expression-P=C}, we have, for any $k\in V$,
		\begin{align*}
			Me_k&
			=\frac{1}{1+w}JD^{-1}e_k-nD^{-1}e_k+\frac{n}{1+w}D^{-1}JD^{-1}e_k\\
			&=\frac{1}{\alpha_k-n-1}\left(\frac{1}{1+w}\mathbf{1}-ne_k+\frac{n}{1+w}D^{-1}\mathbf{1}\right).
		\end{align*}
		
		Then for any distinct $i,j \in V$,
		\begin{align}
			y_{ij}
			&=\mathbf{1}^\top  Me_i\cdot(M_{jj}-1)-\mathbf{1}^\top  Me_j\cdot M_{ji} \nonumber\\
			&=\mathbf{1}^\top  Me_i\cdot(e_j^\top  Me_j-1)-\mathbf{1}^\top  Me_j\cdot e_j^\top Me_i \nonumber\\
			&=\frac{1}{\alpha_i-n-1}\left(\frac{n+1}{1+w}-n+\frac{nw}{1+w}\right)\cdot
			\left[\frac{1}{\alpha_j-n-1}\left(\frac{1}{1+w}-n+\frac{n}{1+w}\cdot\frac{1}{\alpha_j-n-1}\right)-1\right]- \nonumber\\
			&~~~~\frac{1}{\alpha_j-n-1}\left(\frac{n+1}{1+w}-n+\frac{nw}{1+w}\right)\cdot
			\left[\frac{1}{\alpha_i-n-1}\left(\frac{1}{1+w}+\frac{n}{1+w}\cdot\frac{1}{\alpha_j-n-1}\right)\right] \nonumber\\
			&=\frac{-n-(\alpha_j-n-1)}{(\alpha_i-n-1)(\alpha_j-n-1)} \left(\frac{n+1}{1+w}-n+\frac{nw}{1+w}\right) \nonumber\\
			&=\frac{1-\alpha_j}{(\alpha_i-n-1)(\alpha_j-n-1)(1+w)}. \label{eq:yijC}
		\end{align}
		
		Since the denominator is negative, $y_{ij}\leq 0$ with equality holds if and only if $\alpha_j=1$.
		
		Finally, when $\alpha_j=0$, equation \ref{eq:yijC} gives
		\begin{align*}
			y_{ij}
			&=\frac{1-\alpha_j}{(\alpha_i-n-1)(\alpha_j-n-1)(1+w)}\\
			&= \frac{1}{(\alpha_i-n-1)(-n-1)(1+\sum_{k\ne j}\frac{1}{\alpha_k-n-1}+\frac{1}{-n-1})}\\
			&\leq \frac{1}{(-n-1)(-n-1)(1+\frac{n}{1-n-1}+\frac{1}{-n-1})}\\
			&=\frac{1}{(n+1)^2(-\frac{1}{n+1})}\\
			&=-\frac{1}{n+1}
		\end{align*}
		as desired.
	\end{proof}

Lemma~\ref{lemma:claimC} says that $y_{ij}^{(0)} (\alpha_j = 0) \leq -\frac{1}{n+1}$.
The hope is that if $\delta > 0$ is small enough,
then $P^{(\delta)}$ would still be close to $P^{(0)} = C$,
and so $y_{ij}^{(\delta)}(\alpha_j = 0)$ will stay negative.
Hence, we next analyze the quantity $y_{ij}$ as a function of 
the interaction matrix $P$.

	\begin{lemma}
		\label{sum of dyijdPkl}
		Fixing some $\alpha \in [0, 1]^N$, let $A:=\Diag(\alpha)$ and
		let $P$ be a row-stochastic matrix such that $M = M(P):=[I-(I-A)P]^{-1}$ exists. For any distinct $i,j\in V$, denote
		\begin{equation*}
			y_{ij}(P):=\mathbf{1}^\top  Me_i\cdot(M_{jj}-1)-\mathbf{1}^\top  Me_j\cdot M_{ji}
		\end{equation*}
		as a function of $P$. Then,  
		$\displaystyle \sum_{k,l\in N} \left|\frac{\partial y_{ij}}{\partial P_{kl}}\right|\leq 4(\mathbf{1}^\top M\mathbf{1})^3$.
	\end{lemma}
	
	\begin{proof}
		By the definition of the inverse of a matrix $B$, we have
		$BB^{-1} = I$. The partial derivative with respect to a variable $t$ is: $\frac{\partial B}{\partial t}B^{-1}+B\frac{\partial B^{-1}}{\partial t}= 0$. Hence, we have $\frac{\partial B^{-1}}{\partial t}=-B^{-1}\frac{\partial B}{\partial t}B^{-1}$. Applying the above result with $B=I-(I-A)P$ and $t=P_{kl}$ and denoting $M=[I-(I-A)P]^{-1}$, we get
		\begin{equation*}
			\frac{\partial M}{\partial P_{kl}}
			=-M[-(I-A)e_ke_l^\top ]M
			=M(I-A)e_ke_l^\top M
			=(1-\alpha_k)Me_ke_l^\top M.
		\end{equation*}
		Hence, for any $k,l\in N$,
		\begin{align*}
			\left|\frac{\partial y_{ij}}{\partial P_{kl}}\right|
			&=\left|\mathbf{1}^\top (1-\alpha_k)Me_ke_l^\top Me_i(M_{jj}-1)+\mathbf{1}^\top Me_ie_j^\top (1-\alpha_k)Me_ke_l^\top Me_j-\right.\\
			&~~~~\left.\mathbf{1}^\top (1-\alpha_k)Me_ke_l^\top Me_jM_{ji}-\mathbf{1}^\top Me_je_j^\top (1-\alpha_k)Me_ke_l^\top Me_i\right|\\
			&=(1-\alpha_k)\left|\mathbf{1}^\top Me_kM_{li}(M_{jj}-1)+\mathbf{1}^\top Me_iM_{jk}M_{lj}-\mathbf{1}^\top Me_kM_{lj}M_{ji}-\mathbf{1}^\top Me_jM_{jk}M_{li}\right|\\
			&\leq\left|\mathbf{1}^\top Me_kM_{li}(M_{jj}-1)\right|+\left|\mathbf{1}^\top Me_iM_{jk}M_{lj}\right|+\left|\mathbf{1}^\top Me_kM_{lj}M_{ji}\right|+\left|\mathbf{1}^\top Me_jM_{jk}M_{li}\right|\\
			&=\mathbf{1}^\top Me_kM_{li}(M_{jj}-1)+\mathbf{1}^\top Me_iM_{jk}M_{lj}+\mathbf{1}^\top Me_kM_{lj}M_{ji}+\mathbf{1}^\top Me_jM_{jk}M_{li},
		\end{align*}
		where the last inequality holds because $M$ is nonnegative by Fact~\ref{fact:tech}.
		
		Thus, we obtain
		\begin{align*}
			&~~~\sum_{k,l \in N} \left|\frac{\partial y_{ij}}{\partial P_{kl}}\right|\\
			&\leq \sum_{k \in N} \sum_{l \in N} ( \mathbf{1}^\top Me_kM_{li}(M_{jj}-1)+\mathbf{1}^\top Me_iM_{jk}M_{lj}+\mathbf{1}^\top Me_kM_{lj}M_{ji}+\mathbf{1}^\top Me_jM_{jk}M_{li} ) \\
			&=\sum_{k \in N} (\mathbf{1}^\top Me_k\mathbf{1}^\top Me_i(M_{jj}-1)+\mathbf{1}^\top Me_iM_{jk}\mathbf{1}^\top Me_j+\mathbf{1}^\top Me_k\mathbf{1}^\top Me_jM_{ji}+\mathbf{1}^\top Me_jM_{jk}\mathbf{1}^\top Me_i )\\
			&=\mathbf{1}^\top M\mathbf{1}\mathbf{1}^\top Me_i(M_{jj}-1)+\mathbf{1}^\top Me_ie_j^\top M\mathbf{1}\mathbf{1}^\top Me_j+\mathbf{1}^\top M\mathbf{1}\mathbf{1}^\top Me_jM_{ji}+\mathbf{1}^\top Me_je_j^\top M\mathbf{1}\mathbf{1}^\top Me_i\\
			&\leq 4(\mathbf{1}^\top M\mathbf{1})^3
		\end{align*}
		where the last inequality is obtained by applying Fact~\ref{fact:tech} again.
	\end{proof}

In view of Lemma~\ref{sum of dyijdPkl},
we wish to bound the entries of $M^{(\delta)}$
for small $\delta > 0$.
The next fact was given in Alfa et al.\ \cite{Alfa} that gives two-sided bounds to the inverse of a perturbed nonsingular diagonally dominant. 
We use the operator $| \cdot |$ on a matrix to denote the matrix
with the same dimension by taking absolute values entrywise.
	
	\begin{fact}[Entrywise Bounds for Diagonally Dominant Matrix Inverse \cite{Alfa}]\label{inverse-bound}
	Suppose $X$ and $\widetilde{X}$ are matrices 
	of the form $I - B$, where $B \geq 0$ and has spectral norm strictly less than 1, and each row
	of $B$ sums to at most 1.
	
	Let $0\leq \varepsilon <1$ such that $|X_{ij}-\tilde{X}_{ij}| \leq \varepsilon |X_{ij}|$ for $i\ne j$ and $|X \mathbf{1}-\tilde{X} \mathbf{1}| \leq \varepsilon |X\mathbf{1}|$. Then,
		\begin{equation*}
			\frac{(1-\varepsilon)^n}{(1+\varepsilon)^{n-1}}X^{-1} \leq \tilde{X}^{-1} \leq
			\frac{(1+\varepsilon)^n}{(1-\varepsilon)^{n-1}}X^{-1}.
		\end{equation*}
	\end{fact}
	
Recall that $R$ is row-stochastic matrix that represents a normalized 
adjacency matrix of a $d$-regular graph $G=(V,E)$ with the insertion of an isolated vertex~$0$;
also, recall that $C = \frac{1}{n}(J - I)$ represents a clique on $N$.

	\begin{lemma}\label{sum of entries of M}
		Let $\alpha \in [0,1]^N$ such that $\alpha_0 = 1$,
		and $A := \Diag(\alpha)$.
		For $0 \leq \delta < \frac{d}{n}$,
		define $P^{(\delta)} := (1 - \delta) C + \delta R$
		and $M^{(\delta)} := (I - (I-A)P^{(\delta)})^{-1}$.
		Then,
		\begin{equation*}
			\mathbf{1}^\top  M^{(\delta)} \mathbf{1}\leq \displaystyle \frac{n(1+\varepsilon)^n}{(1-\varepsilon)^{n-1}}
			~~~\text{where}~ \varepsilon=\frac{\delta n}{d}.
		\end{equation*}
	\end{lemma}
	\begin{proof}
	For $0 \leq \delta < \frac{d}{n}$,
	$X^{(\delta)} := I - (I-A)P^{(\delta)}$ has the required form
	as stated in Fact~\ref{inverse-bound},
	because $\alpha_0 = 1$ and $P^{(\delta)}$ is irreducible.
	
	We next check that the hypothesis of Fact~\ref{inverse-bound} holds
	for $X^{(0)}$ and $X^{(\delta)}$.
		
		For any distinct $i,j \in V$, since $R_{ij}=0$ or $R_{ij}=\displaystyle \frac1d$, we have
		\begin{align*}
			|X_{i0}^{(0)}-X_{i0}^{(\delta)}|&=\left|-\frac{1-\alpha_i}{n}+\frac{(1-\alpha_i)(1-\delta)}{n}\right| =\frac{\delta(1-\alpha_i)}{n}=\delta |X_{i0}^{(0)}|<\varepsilon|X_{i0}^{(0)}|,\\
			|X_{0j}^{(0)}-X_{0j}^{(\delta)}|&=|0-0| =0 \leq \varepsilon|X_{0j}^{(0)}|,\\
			|X_{ij}^{(0)}-X_{ij}^{(\delta)}|
			&=\left|-\frac{1-\alpha_i}{n}+(1-\alpha_i)\displaystyle \left[\frac{1-\delta}{n}+\delta R_{ij}\right]\right|\\
			&=\left|\delta (1-\alpha_i) \left(R_{ij}-\frac{1}{n}\right)\right|
			\leq \frac{\delta(1-\alpha_i)}{d}=\frac{\delta n}{d}\cdot \frac{1-\alpha_i}{n}=\varepsilon |X_{ij}^{(0)}|.
		\end{align*}
		In addition, as $P^{(\delta)}$ is row-stochastic, we have $P^{(\delta)}\mathbf{1}=\mathbf{1}$ and hence
		\begin{equation*}
			|X^{(0)}\mathbf{1}-X^{(\delta)}\mathbf{1}|
			=|[X^{(0)}-X^{(\delta)}]\mathbf{1}|
			=|\delta (I-A)(P^{(0)}-P^{(\delta)})\mathbf{1}|
			=\mathbf{0}\leq \varepsilon |X^{(0)} \mathbf{1}|.
		\end{equation*}
		
		Therefore, we have $M^{(\delta)}\leq \displaystyle \frac{(1+\varepsilon)^n}{(1-\varepsilon)^{n-1}} M^{(0)}$ 
		by Fact~\ref{inverse-bound}. Then, as $\mathbf{1}^\top  M^{(0)} \mathbf{1}\leq n$ by Lemma \ref{lemma:claimC},
		we have 
		\begin{equation*}
			\mathbf{1}^\top  M^{(\delta)} \mathbf{1}\leq \displaystyle\frac{(1+\varepsilon)^n}{(1-\varepsilon)^{n-1}} \mathbf{1}^\top  M^{(0)} \mathbf{1}\leq \frac{n(1+\varepsilon)^n}{(1-\varepsilon)^{n-1}},
		\end{equation*}
		as desired.
	\end{proof}
	
	With the previous preparation, the next lemma proposes an exact universal perturbation parameter $\delta$ that guarantees the negativity of $y_{ij}$ when $\alpha_j\in [0,1)$, for any distinct $i,j\in V$.
	
	\begin{lemma}
		\label{yij<0}
		Let $\alpha \in [0,1]^N$
		and $M^{(t)}:=[I-(I-A)P^{(t)}]^{-1}$ for $0\leq t<\displaystyle \frac{d}{n}$
		be as defined in Lemma~\ref{sum of entries of M}.
		For any distinct $i,j\in V$, define:
		\begin{equation*}
			y_{ij}^{(t)}:=\mathbf{1}^\top  M^{(t)} e_i \cdot(M_{jj}^{(t)}-1)-\mathbf{1}^\top  M^{(t)} e_j\cdot M_{ji}^{(t)}.
		\end{equation*}
		Let $\displaystyle \delta = \frac{d^3(2d-1)^{3n-3}}{(n+1)^6(2d+1)^{3n}}$. If $\alpha_j\in [0,1)$, then $y_{ij}^{(\delta)} < 0$.
	\end{lemma}
	
	\begin{proof}
		
		Recall that $P^{(t)} := (1-t) \cdot C + t R$ as defined in Lemma~\ref{sum of entries of M}.
		We have:
		
		\begin{equation*}
			P_{ij}^{(t)}=\begin{cases}
				t	& \text{if}~ i=j=0\\
				\frac{1-t}{n}	& \text{if}~ i\ne j ~\text{and either}~ i=0 ~\text{or}~ j=0\\
				0	& \text{if}~ i=j\ne 0\\
				\frac{1-t}{n}+\frac{t}{d} ~\text{or}~ \frac{1-t}{n}	& \text{if}~ i\ne j ~\text{and both}~ i,j\ne0
			\end{cases}.
		\end{equation*}
		
		Hence, differentiating with respect to~$t$, we have:
		\begin{equation*}
			\frac{d}{dt} P_{ij}^{(t)}=\begin{cases}
				1	& \text{if}~ i=j=0\\
				-\frac{1}{n}	& \text{if}~ i\ne j ~\text{and either}~ i=0 ~\text{or}~ j=0\\
				0	& \text{if}~ i=j\ne 0\\
				-\frac{1}{n}+\frac{1}{d} ~\text{or}~ -\frac{1}{n}	& \text{if}~ i\ne j ~\text{and both}~ i,j\ne0
			\end{cases}.
		\end{equation*}

		Next, we treat $P$ as an $(n+1)^2$-dimensional vector, we have $\left\Vert \frac{d}{dt} P^{(t)}\right\Vert_1=\sum_{i\in N}\sum_{j \in N} |\frac{d}{dt} P_{ij}^{(t)}|\leq (n+1)^2$. 
		Moreover, we consider $P^{(t)}$ as a function on~$t$, and recalling
		that $M(P) := [I - (I-A)P]^{-1}$, we can treat the following
		as a function on $P$:
		
		$$y_{ij}(P) := \mathbf{1}^\top  M(P) e_i \cdot(M_{jj}(P)-1)-\mathbf{1}^\top  M(P) e_j\cdot M_{ji}(P).$$
		
		Hence, we can also treat $\nabla_P \, y_{ij}(P)$ as an 
		$(n+1)^2$-dimensional vector, and use $\langle \cdot, \cdot \rangle$
		to denote the corresponding inner product operation.
	  We have the following.

		\begin{align*}
			|y_{ij}^{(\delta)}-y_{ij}^{(0)}|
			&= \left|\int_0^\delta \langle \nabla_P \, y_{ij}(P^{(t)}) , \frac{d}{dt} P^{(t)} \rangle ~dt\right| & \text{(Fundamental theorem for line integrals)}\\
			&\leq\int_0^\delta \left|\langle \nabla_P \, y_{ij}(P^{(t)}) , \frac{d}{dt} P^{(t)} \rangle\right| ~dt\\
			&\leq \int_0^\delta \left\Vert \nabla_P \, y_{ij}(P^{(t)}) \right\Vert_2 \cdot \left\Vert P'(t)\right\Vert_2 ~dt & \text{(Cauchy-Schwarz Inequality)}\\
			&\leq \int_0^\delta \left\Vert \nabla_P \, y_{ij}(P^{(t)}) \right\Vert_1\cdot \left\Vert P'(t)\right\Vert_1 ~dt\\
			&\leq \int_0^\delta \left[\sum_{k \in N} \sum_{l \in N} \left|\frac{\partial y_{ij}(P^{(t)})}{\partial P_{kl}}\right|\right]\cdot (n+1)^2 ~dt\\
			&\leq \int_0^\delta 4(\mathbf{1}^\top M^{(t)} \mathbf{1})^3\cdot (n+1)^2 ~dt & \text{(Lemma \ref{sum of dyijdPkl})}\\
			&\leq \int_0^\delta 4\left[\frac{n(1+\varepsilon)^n}{(1-\varepsilon)^{n-1}}\right]^3\cdot (n+1)^2 ~dt & \text{(Lemma \ref{sum of entries of M})}\\
			&\leq \frac{4\delta (n+1)^5 (1+\varepsilon)^{3n}}{(1-\varepsilon)^{3n-3}}
		\end{align*}
		where $\displaystyle \varepsilon=\frac{\delta n}{d}$.
		Since $\displaystyle \delta  = \frac{d^3(2d-1)^{3n-3}}{(n+1)^6(2d+1)^{3n}} <\frac1{2n}$, $\varepsilon=\displaystyle \frac{\delta n}{d}<\frac1{2d}$. Hence, we have
		\begin{equation*}
			|y_{ij}(\delta)-y_{ij}(0)|
			<\frac{4\delta (n+1)^5 (1+\frac1{2d})^{3n}}{(1-\frac1{2d})^{3n-3}}
			=\frac{\delta (n+1)^5 (2d+1)^{3n}}{2d^3(2d-1)^{3n-3}}.
		\end{equation*}
		
		Also, since $\displaystyle \delta = \frac{d^3(2d-1)^{3n-3}}{(n+1)^6(2d+1)^{3n}}$, we have
		\begin{equation*}
			|y_{ij}^{(\delta)}-y_{ij}^{(0)}|
			< \frac{(n+1)^5 (2d+1)^{3n}}{2d^3(2d-1)^{3n-3}}\cdot \frac{d^3(2d-1)^{3n-3}}{(n+1)^6(2d+1)^{3n}}
			=\frac{1}{2(n+1)},
		\end{equation*}
		
		for any $\alpha \in [0,1]^N$ with $\alpha_0=1$ and $\alpha_k\in[0,1]$ for any other $k \in V$. In particular, this is true when $\alpha_j = 0$.
	From Lemma \ref{lemma:claimC}, $y_{ij}^{(0)} \vert_{\alpha_j = 0} \leq \displaystyle -\frac{1}{n+1}$. Hence, $y_{ij}^{(\delta)}\vert_{\alpha_j = 0} < \displaystyle -\frac{1}{2(n+1)}<0$. Then by Lemma \ref{monotone}, $y_{ij}^{(\delta)}<0$ for any $\alpha_j\in [0,1)$.
	\end{proof}

Finally, we are ready to prove the main result of this section.

\begin{proofof}{Lemma~\ref{lemma:structure}}

As mentioned before, the formal goal is to show
statement~(\ref{eq:Hessian}).
Lemma~\ref{dirderi} says that 
it suffices to show that $y_{ij}(\alpha) < 0$ when $\alpha_j < 1$.
Finally,
Lemma~\ref{yij<0} says that 
by choosing $\delta =  \frac{d^3(2d-1)^{3n-3}}{(n+1)^6(2d+1)^{3n}}$
to define $P^{(\delta)} = (1 - \delta) C + \delta R$,
we have $y^{(\delta)}_{ij}(\alpha) < 0$ for $\alpha_j < 1$, as required.
\end{proofof}

\bibliographystyle{alpha}
\bibliography{mybib,opinion,dihyper}

\end{document}